\documentclass[preprint,12pt]{elsarticle}

\usepackage{graphicx}%
\usepackage{multirow}%
\usepackage{amsmath,amssymb,amsfonts}%
\usepackage{amsthm}%
\usepackage{mathrsfs}%
\usepackage[title]{appendix}%
\usepackage{xcolor}%
\usepackage{textcomp}%
\usepackage{manyfoot}%
\usepackage{booktabs}%
\usepackage{algorithm}%
\usepackage{algorithmicx}%
\usepackage{algpseudocode}%
\usepackage{listings}%
\usepackage[justification=centering]{caption}%
\usepackage{hyperref}
\biboptions{sort&compress}

\usepackage{mathtools}
\usepackage{tikz}
\usepackage{pgf}
\usetikzlibrary{arrows,snakes,backgrounds}
\usetikzlibrary{arrows.meta, chains,positioning,shapes.geometric}
\usetikzlibrary{decorations.pathreplacing,decorations.markings}
\usetikzlibrary{positioning,automata}
\usepackage{tkz-graph}

\newtheorem{atheorem}{Theorem}

\theoremstyle{thmstyleone}%
\newtheorem{theorem}{Theorem}
\newtheorem{proposition}[theorem]{Proposition}%

\theoremstyle{thmstyletwo}%

\theoremstyle{thmstylethree}%
\newtheorem{definition}{Definition}%

\journal{Journal}

\begin{document}

\begin{frontmatter}



\title{The Price of Cognition and Replicator Equations in Parallel Neural Networks}


\author{Armen Bagdasaryan\corref{mycorrespondingauthor}}
\cortext[mycorrespondingauthor]{Corresponding author}
\ead{armen.bagdasaryan@aum.edu.kw}

\author{Antonios Kalampakas}

\author{Mansoor Saburov}

\affiliation{organization={College of Engineering and Technology, American University of the Middle East}, 
	city={Egaila},
	postcode={54200}, 
	country={Kuwait}}

\begin{abstract}
In a broader context, the experimental investigation postulates that toxic neuropeptides/neurotransmitters are causing damage to the functionality of synapses during neurotransmission processes. In this paper, we are aiming to propose a novel mathematical model that studies the dynamics of synaptic damage in terms of concentrations of toxic neuropeptides/neurotransmitters during neurotransmission processes. Our primary objective is to employ Wardrop's first and second principles within a neural network of the brain. In order to comprehensively incorporate Wardrop's first and second principles into the neural network of the brain, we introduce two novel concepts: \textit{neuropeptide's (neurotransmitter's) equilibrium} and \textit{synapses optimum}. The \textit{neuropeptide/neurotransmitter equilibrium} refers to \textit{a distribution of toxic neuropeptides/neurotransmitters that leads to uniform damage across all synaptic links}. Meanwhile, \textit{synapses optimum} is \textit{the most desirable distribution of toxic neuropeptides/neurotransmitters that minimizes the cumulative damage experienced by all synapses}. In the context of a neural network within the brain, an analogue of the price of anarchy is \textit{the price of cognition} which is \textit{the most unfavorable ratio between the overall impairment caused by toxic neuropeptide's (neurotransmitter's) equilibrium in comparison to the optimal state of synapses (synapses optimum)}. To put it differently, \textit{the price of cognition} measures \textit{the loss of cognitive ability resulting from increased concentrations of toxic neuropeptides/neurotransmitters}. Additionally, a replicator equation is proposed within this framework that leads to the establishment of the synapses optimum during the neurotransmission process. It is important to note that our model serves as a high-level simplification and abstraction of the natural neurotransmission process involving interactions between two neurons. Nevertheless, we envision that this mathematically abstract model can serve as a source of motivation to instigate novel experimental, mathematical, and computational research avenues in the field of contemporary neuroscience.
\end{abstract}


\begin{highlights}
\item The connection between two neurons, established through unidirectional synaptic links, is visualized as \textit{a parallel network between two nodes}; 
\item Our innovative approach involves the application of ``\textit{Wardrop's first and second principles}'' in the examination of synaptic damage resulting from toxic neuropeptides/neurotransmitters during neurotransmission processes; 
\item This was achieved by introducing two novel concepts: \textit{neuropeptide's (neurotransmitter's) equilibrium} and \textit{synapses optimum}; 
\item The \textit{neuropeptide/neurotransmitter equilibrium} refers to {a distribution of toxic neuropeptides/neurotransmitters that leads to uniform damage across all synaptic links}; 
\item \textit{Synapses optimum} is the most desirable distribution of toxic neuropeptides/neurotransmitters that minimizes the cumulative damage experienced by all synapses;
\item An analogous concept  to the price of anarchy is \textit{the price of cognition}  that represents \textit{the most unfavorable ratio between the overall impairment caused by toxic neuropeptide's (neurotransmitter's) equilibrium in comparison to the optimal state of synapses (synapses optimum)};  
\item To put it differently, \textit{the price of cognition} measures \textit{the loss of cognitive ability resulting from increased concentrations of toxic neuropeptides/neurotransmitters};  
\item Finally, a replicator equation is proposed within this framework that leads to the establishment of the synapses optimum during the neurotransmission process.
\end{highlights}


\end{frontmatter}


\section{Introduction}\label{sec1}

The brain assumes responsibility for each cognitive process such as thou\-ghts, emotions, and physical actions. This role is facilitated through a fundamental physiological mechanism known as \textit{neurotransmission}. In essence, \textit{neurotransmission} is the intricate process through which brain cells establish communication channels. Predominantly, this intricate exchange of information transpires at a specialized inter-cellular junction known as the \textit{synapse}. It is now well-established within the domain of contemporary neuroscience that the synapse plays a critical role in a variety of cognitive processes, especially those involved with deep learning and memory. The term ``\textit{synapse}'' finds its etymological origins in the Greek words ``\textit{syn}'' denoting ``\textit{together},'' and ``\textit{haptein}'' signifying to  ``\textit{fasten}.'' The \textit{synapse} is the small pocket of space situated between two neurons which is essential to the transmission of electric nerve impulse from one neuron to another. In the context of a chemical synapse, the membrane of the transmitting neuron, often referred to as the \textit{presynaptic} neuron, closely approximates the membrane of the recipient neuron, known as the \textit{postsynaptic} neuron. In both the presynaptic and postsynaptic neurons, there exist comprehensive arrays of molecular machinery that \textit{link} the two membranes together and carry out the signal transmission process. A single  neuron has the capacity to house thousands of \textit{synaptic links}. The linkage space between a presynaptic fiber and a postsynaptic fiber is called the \textit{synaptic cleft}. The presynaptic neuron releases neurotransmitters from synaptic vesicles into the synaptic cleft. Then those neurotransmitters will be taken up by the membrane of receptors on the postsynaptic neuron throughout the \textit{unidirectional synaptic links}. In essence, neurotransmitters are essential neurochemicals that maintain synaptic and cognitive functions in humans by sending signals across presynaptic to postsynaptic neurons throughout the \textit{unidirectional synaptic links}. The reader can refer to the monograph (see \cite{Della2021}) for the detailed discussion of elements of contemporary neuroscience.

In the past, neuroscientists held the belief that all synapses remained \textit{constant} in their functionality, operating consistently at a \textit{uniform} level. However, contemporary understanding has evolved to acknowledge (see \cite{FWT2015}) that synaptic \textit{strength} can be modified by activity or inactivity, leading to strengthening or weakening of synapses or even \textit{damaging} the functionality of synapses in the brain. The augmentation of synaptic strength occurs in direct proportion to its usage, 
thereby enhancing its capacity to exert a more substantial influence over its adjacent postsynaptic neurons. Contemporary neuroscientists believe (see \cite{Della2021,FWT2015}) that this strengthening of synapses constitutes 
a fundamental mechanism for facilitating the process of learning and, 
as a result, the formation of memories.

Gradually, individuals afflicted with \textit{Alzheimer's disease} undergo a profound and relentless decline in memory and cognitive abilities. These transformations stem from the gradual \textit{deterioration} and \textit{dysfunction} of \textit{synapses} which are responsible for the encoding, retention, and processing of information. In 2017, the issue 4 of volume 57 of the \textit{Journal of Alzheimer's Disease} was specially dedicated to critically assess the current status of synapses and neurotransmitters in Alzheimer’s disease. Experts in the domains of synapses and neurotransmitters of Alzheimer's disease provide a comprehensive overview of the present state of fundamental biology of synapses and neurotransmitters (see \cite{Reddy2017}). Additionally, they offer an update on the current progress of clinical trials involving neurotransmitters for the treatment of Alzheimer's disease. There is a growing body of evidence (see \cite{Reddy2017} and the references provided therein) indicating that impairments in the cortex and hippocampus at the initial stages of Alzheimer's disease are linked to synaptic \textit{damage} and \textit{deterioration} which are attributed to the presence of oligomeric forms of the neurotoxic amyloid-$\beta$ (A$\beta$) peptide acting as neurotransmitter. According to 2016 statistics from the Alzheimer’s Association (see \cite{RajmohanReddy2017,JJKSSAK2017,WR2017,KR2017,GTD2017,CT2017,TT2017,JB2017,GS2017,ECMGGB2017}), recent molecular, cellular, animal model, and postmortem brain studies have revealed that amyloid-$\beta$ (A$\beta$) is a key factor that causes synaptic damage and cognitive decline in patients with Alzheimer’s disease. In this context, it has been postulated (see  \cite{hely1, RajmohanReddy2017,JJKSSAK2017,WR2017,KR2017,GTD2017,CT2017,TT2017,JB2017,GS2017,ECMGGB2017}) that elevated levels of amyloid-$\beta$ (A$\beta$) may induce early \textit{defects} and \textit{damage} in synaptic function. This has been proposed as a potential contributing factor to the manifestation of neuropsychiatric symptoms commonly observed during the prodromal phase of Alzheimer's disease. The reader may refer to the special issue 57(4):(2017) of the \textit{Journal of Alzheimer's Disease} for a comprehensive report of the present state of fundamental biology of synapses and neurotransmitters in the field of Alzheimer’s disease (see also \cite{hely2, CWHDZZ2024,PHLPS2023,HYKDPP2022,KK2014,LKP2023,HKV2023}).

Due to the importance of knowing \textit{damage} and \textit{deterioration} of synapses in the brains of people with Alzheimer’s disease, in this paper, we are aiming to propose a novel mathematical model that \textit{studies the dynamics of synaptic damage in terms of concentrations of toxic neuropeptides/neurotransmitters during neurotransmission processes}.  Our objective is to employ ``\textit{Wardrop's first and second principles}'' in order to \textit{enhance our understanding of the dynamics of toxic neuropeptides/neurotransmitters during the neurotransmission process within a neural network of the brain}. We first introduce the so-called \textit{neuropeptide's (neurotransmitter's) equilibrium} concept that encapsulates the fundamental essence of Wardrop's first principle. Namely, the \textit{neuropeptide/neurotransmitter equilibrium} refers to \textit{a distribution of toxic neuropeptides/neurotransmitters that leads to uniform damage across all synaptic links}. Subsequently, we proceed to present a complete manifestation of Wardrop's second principle through the introduction of another so-called \textit{synapses optimum} concept. \textit{Synapses optimum} is \textit{the most desirable distribution of toxic neuropeptides/neurotransmitters that minimizes the cumulative damage experienced by all synapses}.
Consequently, within the context of a neural network of the brain, an analogue of the price of anarchy is \textit{the price of cognition} which is \textit{the most unfavorable ratio between the overall impairment caused by toxic neuropeptide's (neurotransmitter's) equilibrium in comparison to the optimal state of synapses (synapses optimum)}. To put it differently, \textit{the price of cognition} measures \textit{the loss of cognitive ability resulting from increased concentrations of toxic neuropeptides/neurotransmitters}. Finally, a replicator equation is proposed within this framework that leads to the establishment of the synapses optimum during the neurotransmission process.  Our model represents a sophisticated simplification and abstraction  of the natural neurotransmission process involving two neurons. Nevertheless, we envision that this mathematically abstract model can serve as a catalyst and a source of motivation to instigate novel experimental, mathematical, and computational research avenues in the field of contemporary neuroscience.

\section{The Price of Cognition Through Wardrop's Equilibria Approaches}\label{sec2}

A \textit{brain neural network} consists of a collection of neurons that are \textit{chemically interconnected} or \textit{functionally associated}. Within the brain neural network, \textit{synapses} play a vital role in facilitating the transmission of nervous impulses from one neuron to another. The fundamental types of connections between neurons include both \textit{chemical} and \textit{electrical synapses}. An \textit{electrical} \textit{synapse}'s primary advantage lies in its ability to swiftly transmit signals from one neuron to the adjacent neuron. In the context of \textit{chemical synapses}, it's important to note that there is a temporal delay in the neurotransmission process. Therefore, a part of the brain neural network that relies on chemical synapses for communication can be regarded as a \textit{transportation network}. Within this framework, the \textit{impairment}, \textit{deterioration}, \textit{degradation}, or \textit{decline} of the functionality of synapses, occurring due to elevated levels of toxic neuropeptides/neurotransmitters during the neurotransmission process, can be linked to a \textit{traffic congestion} problem in a transportation network. From a theoretical standpoint, when investigating the dynamic relationship between synaptic damage and the concentration of toxic neuropeptides/neurotransmitters during the neurotransmission process, it is plausible to adapt ``\textit{Wardrop's first and second principles}''. Although these principles were originally formulated to describe \textit{optimal flow} distributions in transportation networks, we believe that they can be adapted and applied effectively within the framework of the neural network in the brain. This is a novel contribution of this paper.

The network optimization problems have attracted much attention over the last decade for their ubiquitous appearance in real-life applications and the inherent mathematical challenges that they present, especially, in optimal transportation theory and communication networks (see \cite{W1952,Beckmann1956,daf-disser,DS1969,Patriksson,AcemogluOzdaglar,OzdaglarSrikant}). Back in 1952, J.G.Wardrop (see \cite{W1952}) formulated two principles of optimality of flows in networks that describe the circumstances of the \textit{user equilibrium} and the \textit{system optimum.} The first Wardrop principle states 
that \textit{the costs of all utilized links are equal and less than the costs of those unutilized links for every fixed source-destination pair}. Meanwhile, the \textit{system optimum} is \textit{the optimal distribution of the flow for which the total cost of the system is minimal}.   
The problems of finding the user equilibrium and system optimum are the topics of active research both in theory and practice. When making route choices in traffic networks, the network users frequently display selfish behavior, that is the fundamental principle of the first Wardrop principle, by selecting routes that minimize their individual travel costs. It is widely recognized (see \cite{Beckmann1956,daf-disser,DS1969,Patriksson}) that selfish routing, in general, does not lead to a system optimum of the network that minimizes the total travel cost. The so-called \textit{Price of Anarchy} (see \cite{KP1999,KP2009}) is a quantitative measure of the inefficiency of the traffic network that was caused by the selfish behavior of the network users. Namely, \textit{the price of anarchy} is \textit{the the worst-possible ratio between the total cost of a user equilibrium in comparison to the system optimum.}

The primary objective of this paper is to gain insight into specific mechanisms and phenomena within the neural network of the brain by applying methods originally designed for solving traffic assignment problems in transportation networks. For the sake of simplicity, we focus on examining the process of signal transmission between two neurons through chemical synapses.

Broadly speaking, there exist comprehensive \textit{arrays} (\textit{links}) of \textit{synapses} between presynaptic and postsynaptic neurons that carry out the signal transmission process. In this scenario, let's assume that the presynaptic neuron releases \textit{one-unit} of toxic \textit{neuropeptides}/\textit{neurotransmitters}. These toxic \textit{neuropeptides}/\textit{neurotransmitters} are subsequently absorbed by receptors located on the postsynaptic neuron across \textit{unidirectional synaptic links}. We can visualize \textit{two neurons connected by unidirectional synaptic links} as \textit{a parallel network between two nodes}.
Let $\mathbf{I}_m=\{1,2,\dots,m\}$ denote the set of the synaptic links between two neurons. We denote by $x_k\geq 0$, the concentration of toxic \textit{neuropeptides}/\textit{neurotransmitters} passing through the synaptic link $k$  and $\mathbf{x}=(x_1,x_2,\ldots,x_m)$ denote a distribution of toxic \textit{neuropeptides}/\textit{neurotransmitters} over all synaptic links $\mathbf{I}_m=\{1,2,\dots,m\}$, where $\sum_{i=1}^m  x_i=1$.  

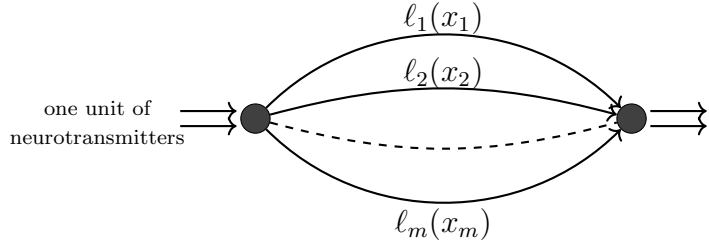
\begin{figure}\label{fig1}
	\begin{center}
		\usetikzlibrary {graphs} 
		\begin{tikzpicture}
			\node [circle,draw,fill=black!75, minimum size=0.33cm] (U1) at (1,2) {};
			\node (L1) at (3.5,3.34) {$\ell_1(x_1)$};
			\node [circle,draw,fill=black!75, minimum size=0.33cm] (U2) at (6,2) {};
			\node (L2) at (3.5,2.6) {$\ell_2(x_2)$};
			\node (L3) at (3.5,0.6) {$\ell_m(x_m)$};
			\node (L4) at (-1.1,2.15) {\scriptsize{one unit of}};
			\node (L4) at (-1.1,1.75) {\scriptsize{neurotransmitters}};
			\Edge[style={->, bend right=45}](U1)(U2);
			\draw[thick, ->] (0,2.1) to (0.75,2.1);
			\draw[thick, ->] (0,1.9) to (0.75,1.9);
			\draw[thick, ->] (6.25,2.1) to (7,2.1);
			\draw[thick, ->] (6.25,1.9) to (7,1.9);
			\Edge[style={-, dashed, bend right=15}](U1)(U2);
			\Edge[style={->, bend left=15}](U1)(U2);
			\Edge[style={->, bend left=45}](U1)(U2)
		\end{tikzpicture}
		\caption{The model of neurotransmission process between two neurons}
	\end{center}
\end{figure}

Within the framework of the neurotransmission process, increased levels of toxic neuropeptides/neurotransmitters damage the synaptic links. To investigate this correlation, we make the assumption that each synaptic link $k\in \mathbf{I}_m$ is associated with a \textit{synaptic damage function} $
\ell_k :[0,\infty) \to [0,\infty )$ dependent on the concentration of toxic neuropeptides/neurotransmitters that measures the damage on the synaptic link $k$. Throughout this paper, for each $k\in\mathbf{I}_m$ the synaptic damage function $\ell_k$ will be assumed to be a convex, strictly increasing, continuously differentiable function and $\ell_k(0)=0$. Let $\mathbf{L}_m(\mathbf{x})=\left(\ell_1(x_1),\ldots, \ell_m(x_m)\right)$ denote a \textit{synaptic damage vector function} at  $\mathbf{x}=(x_1,\ldots,x_m)$. Hence, a synaptic damage between two neurons can be identified with the synaptic damage vector function   $\mathbf{L}_m(\cdot)=(\ell_1(\cdot),\cdots , \ell_m(\cdot))$. In this paper, we will  interchangeably use a synaptic damage vector function and the neural network involving two neurons. We denote by $\mathbb{S}^{m-1}=\{\mathbf{x}\in\mathbb{R}_+^m:\sum_{i=1}^m  x_i =1\}$ the standard simplex. Define $\textup{\bf{supp}}(\mathbf{x}):=\{i\in \mathbf{I}_m:x_i\neq 0\}$ and  $\mathrm{int}\mathbb{S}^{m-1}=\{x\in\mathbb{S}^{m-1}:\textup{\bf{supp}}(\mathbf{x})=\mathbf{I}_m\}$.

A distribution $\mathbf{x}=(x_1,\dots , x_m)\in\mathbb{S}^{m-1}$ is called  a \textit{neuropeptide/neurotransmitter equilibrium} if the synaptic damage is the same across all synaptic links, i.e., 
$\ell_i(x_i)=\ell_j(x_j)$ for all $i,j\in \mathbf{I}_m$. The \textit{cumulative damage} experienced by synapses at a distribution $\mathbf{x}=(x_1,\dots , x_m)\in\mathbb{S}^{m-1}$ is defined by the sum $\sum_{i=1}^m  x_i \ell_i(x_i)$. Obviously, the \textit{neuropeptide/neurotransmitter equilibrium} that results in uniform damage across all synaptic links does not generally minimize the cumulative damage experienced by all synapses. The primary objective of experts in the domains of synapses and neurotransmitters of Alzheimer's disease is \textit{to reduce and minimize  the cumulative damage experienced by synapses in the brains of individuals afflicted with Alzheimer's disease}. Therefore, a distribution $\mathbf{x}=(x_1,\dots ,x_m)\in \mathbb{S}^{m-1}$ is called the \textit{optimal state of synapses} (\textit{synapses optimum}) if it minimizes the cumulative damage $\sum_{i=1}^m  x_i \ell_i(x_i)$ experienced by all synapses.

It has been well-established (see \cite{Beckmann1956,DS1969,Patriksson}) that if for each $k\in\mathbf{I}_m$ a synaptic damage function $\ell_k :[0,\infty) \to [0,\infty )$ with  $\ell_k(0)=0$ is a convex, strictly increasing, continuously differentiable function then there always exists a unique neuropeptide/neurotransmitter equilibrium $\mathbf{x}^{(\mathbf{ne})}\in\mathrm{int}\mathbb{S}^{m-1}$, as well as a unique synapses optimum $\mathbf{x}^{(\mathbf{so})}\in\mathrm{int}\mathbb{S}^{m-1}$. In general, we have that $\mathbf{x}^{(\mathbf{ne})}\neq\mathbf{x}^{(\mathbf{so})}$.

Consequently, within the context of a neural network of the brain, an analogous concept  to \textit{the price of anarchy} can be called as \textit{the price of cognition}. This represents \textit{the most unfavorable ratio between the overall impairment caused by toxic neuropeptide's (neurotransmitter's) equilibrium in comparison to the optimal state of synapses (synapses optimum)}. To put it differently, \textit{the price of cognition} measures \textit{the loss of cognitive ability resulting from increased concentrations of toxic neuropeptides/neurotransmitters}. In terms of the mathematical formula, \textit{the price of cognition} is defined as follows
$$
\mathbf{PoC}({\mathbf{L}_m}):=\frac{\mathbf{CD}(\mathbf{x}^{(\mathbf{ne})})}{\mathbf{CD}(\mathbf{x}^{(\mathbf{so})})}
$$ 
where $\mathbf{x}^{(\mathbf{ne})}\in\mathrm{int}\mathbb{S}^{m-1}$ is the unique neuropeptide/neurotransmitter equilibrium, $\mathbf{x}^{(\mathbf{so})}\in\mathrm{int}\mathbb{S}^{m-1}$ is the unique synapses optimum, and $\mathbf{CD}(\mathbf{x})=\sum_{i=1}^m  x_i \ell_i(x_i)$ is the cumulative damage experienced by all synapses at a distribution  $\mathbf{x}\in\mathbb{S}^{m-1}$. 

Obviously, we always have that $\mathbf{PoC}({\mathbf{L}_m})\geq 1$. The case $\mathbf{PoC}({\mathbf{L}_m})=1$ represents the most favorable and desirable circumstance for individuals afflicted with Alzheimer's disease. Indeed, if  $\mathbf{PoC}({\mathbf{L}_m})=1$ then we must have that $\mathbf{x}^{(\mathbf{ne})}=\mathbf{x}^{(\mathbf{so})}$ due to the uniqueness of the {neuropeptide/neurotransmitter equilibrium} as well as synapses optimum. This means that \textit{the  neuropeptide/neurotransmitter equilibrium that leads to uniform damage across all synaptic links is also the optimal state of synapses which minimizes the cumulative damage experiences by them.} Under such circumstances, when two equilibria coincide $\mathbf{x}^{(\mathbf{ne})}=\mathbf{x}^{(\mathbf{so})}$ then  it is referred to as  \textit{the optimal state of neurotransmitters-synapses (\textit{neurotransmitters-synapses optimum})}. 

Within the domain of transportation networks, the significance of achieving a \textit{price of anarchy} of $1$ in a transportation network was emphasized by S.~Dafermos as early as 1968, even before the concept of \textit{the price of anarchy} itself was formally introduced. Namely, the problem proposed by S.~Dafermos in her Ph.D. thesis \cite{daf-disser} (see also \cite{DS1969}) seeks to identify the specific class of cost functions for a given network that lead to the emergence of identical equilibria according to Wardrop's first and second principles. According to S.~Dafermos, ``\textit{such networks are extremely desirable because the pattern created by the individuals acting in their own self interests coincides with the pattern most economical for the total society} (see \cite{daf-disser,DS1969}).  She also provided a solution to her problem by specifying a family of cost functions in the form of monomial polynomials $\ell_i(x_i)=a_ix_i^k$ where $a_i>0$ and $k\in\mathbb{N}$. This family of cost functions is capable of solving the problem for any given network (see \cite{daf-disser,DS1969}). 

Recent empirical studies in real-world networks (see \cite{price}), as well as theoretical studies (see \cite{self1,self2}), show that the actual value of the price of anarchy is very close to $1$.
The price of anarchy in proximity to $1$ indicates that \textit{the user equilibrium is approximately socially optimal, thereby implying that the effects of selfish behavior are relatively benign}.
Moreover, theoretical studies on real-world networks also indicate (see \cite{self1,self2}) that \textit{the price of anarchy is approximately $1$ in both light and heavy traffic conditions}. In addition, the price of anarchy for a broad range of cost functions, which includes all polynomials, tends to approach $1$ in both heavy and light traffic, regardless of the network topology and the number of origin and destination pairs within the network. Particularly, \textit{the comprehensive asymptotic principle}  for polynomial cost functions states (see \cite{self1,self2}) that in networks where costs are polynomial, \textit{the price of anarchy converges to $1$ under both light and heavy traffic}. This implies that in light and heavy traffic congestion conditions, \textit{a benevolent social planner with complete authority over traffic assignment would not do any better than the selfish behavior of users}. 

On the other hand, parallel networks hold a unique status among other network types. Namely, \textit{parallel networks are  nontrivial classes of network topologies for which the price of anarchy is smaller than for any other networks} (see \cite{RT2003}). Among many other important results, one of the main results of the paper (see \cite{RT2003}) state that under weak assumptions on the class of cost functions, the price of anarchy for any multicommodity flow network is achieved  by a \textit{single-commodity instance} within a \textit{network of parallel links}. Consequently, for a given class of cost functions, the price of anarchy cannot be improved by any nontrivial constraint on the class of network topologies and/or the number of commodities. Therefore, \textit{the lower value of the price of anarchy can always be attained by the simplest parallel network} (see \cite{RT2003}). 

Despite its profound importance and practical applicability (see \cite{price,self1,self2}), Dafermos's problem did not receive adequate attention from experts until more recently. All of these theoretical discoveries and empirical observations (see \cite{self1,self2,price}) served as the driving force behind the exploration of the  central and fundamental questions addressed in our recent studies(see \cite{lect,ccis-dyn,ccis-stat}): 

\begin{itemize}
	\item Can these observations be justified theoretically?
	\item Can we theoretically describe all cost functions associated with parallel networks in which the price of anarchy achieves its least value of 1?
\end{itemize}

Finally, Dafermos's problem has been fully solved for parallel networks in the papers (see \cite{lect,ccis-dyn,ccis-stat}). For any prior given positive distribution $\mathbf{p}=(p_1,\ldots,p_m)\in\mathrm{int}\mathbb{S}^{m-1}$ and for any convex, strictly increasing, continuously differentiable function $\ell:[0,\infty) \to [0,\infty)$ with $\ell(0)=0$, we define the following  cost vector function
$$\mathbf{\Phi}_m(\mathbf{x})=\left(\ell\left(\frac{x_1}{p_1}\right),\ldots, \ell\left(\frac{x_m}{p_m}\right)\right).$$ 

Among many other interesting results (see \cite{lect,ccis-dyn,ccis-stat}), one of the main results  states that the distribution $\mathbf{p}=(p_1,\ldots,p_m)\in\mathrm{int}\mathbb{S}^{m-1}$ is the  \textit{user equilibrium} as well as the \textit{system optimum} of the parallel network $\mathbf{\Phi}_m$. In this case, the \textit{price of anarchy} is equal to its least value that is $1$. This shows that there are a vast class of parallel networks for which the \textit{price of anarchy} is always equal to its least value that is $1$. Unlike Dafermos's example, there are also another interesting classes of parallel networks $\mathbf{L}_m(\mathbf{x})=\left(\ell_1(x_1),\ldots, \ell_m(x_m)\right)$ generated by different degrees of polynomials for which the \textit{price of anarchy} is always equal to $1$, where
\begin{align*}
	& {\ell}_1(x_1)=\frac{x_1}{p_1}, \\
	& {\ell}_2(x_2) = \frac{a_{22}}{2}\Big(\frac{x_2}{p_2}\Big)^2 + \frac{a_{21}}{1}\Big(\frac{x_2}{p_2}\Big) + \frac{a_{22}}{2}, \\
	& {\ell}_3(x_3) = \frac{a_{33}}{3}\Big(\frac{x_3}{p_3}\Big)^3 + \frac{a_{32}}{2}\Big(\frac{x_3}{p_3}\Big)^2 + \frac{a_{31}}{1}\Big(\frac{x_3}{p_3}\Big) + \frac{2a_{33}}{3} + \frac{a_{32}}{2}\\
	&\vdots \\
	&\ell_i(x_i) = \sum_{k=1}^i \frac{a_{ik}}{k}\Big(\frac{x_i}{p_i}\Big)^k + \sum_{k=2}^i \frac{k-1}{k}a_{ik} \\
	& \vdots \\
	& {\ell}_m(x_m) = \sum_{k=1}^m \frac{a_{mk}}{k}\Big(\frac{x_m}{p_m}\Big)^k + \sum_{k=2}^m \frac{k-1}{k}a_{mk}.
\end{align*} 
and 
\[
\sum_{k=1}^{i}a_{ik}=1 \quad \textup{and} \quad  a_{ik}>0 \quad \textup{for all} \quad 1\leq i\leq m, \ 1\leq k\leq i.
\]

\section{Replicator Equations in Parallel Neural Networks}

Let us consider the following parallel neural network
$$\mathbf{\Phi}_m(\mathbf{x})=\left(\ell\left(\frac{x_1}{p_1}\right),\ldots, \ell\left(\frac{x_m}{p_m}\right)\right).$$ 
where $\ell:[0,\infty) \to [0,\infty)$ is a convex, strictly increasing, continuously differentiable function such that $\ell(0)=0$. In this case, as we already mentioned that $\mathbf{p}=(p_1,\ldots,p_m)\in\mathrm{int}\mathbb{S}^{m-1}$ is  the \textit{neurotransmitters-synapses optimum}, i.e., it is simultaneously  the \textit{neuropeptide/neurotransmitter equilibrium} as well as the \textit{synapses optimum} of the parallel network $\mathbf{\Phi}_m$.

Let $\langle\mathbf{x},\mathbf{\Phi}_m(\frac{\mathbf{x}}{\mathbf{p}})\rangle=\sum_{k=1}^mx_k{\ell}\left(\frac{x_k}{p_k}\right)$ be the cumulative damage experienced by all synapses of the parallel neural network $\mathbf{\Phi}_m$ at the distribution $\mathbf{x}\in\mathbb{S}^{m-1}$. 

We assume that the parallel neural network may \textit{dynamically} change the distribution of toxic neuropeptides/neurotransmitters over alternative parallel synaptic links. Namely, if $\mathbf{x}^{(n)}=(x_1^{(n)},\ldots,x_m^{(n)})\in\mathbb{S}^{m-1}$ is the distribution at the step~$n$ then the relative growth rate $\frac{x^{(n+1)}_k-x^{(n)}_k}{x^{(n)}_k}$ of the distribution on the synaptic link $k$ at the step~$(n+1)$ is negatively proportional (with the propositional coefficient $\varepsilon$) to the difference between the damage ${\ell}(\frac{x^{(n)}_k}{p_k})$ of the synaptic link~$k$ at the step~$n$ and the cumulative damage  $\langle\mathbf{x},\mathbf{\Phi}_m(\frac{\mathbf{x}}{\mathbf{p}})\rangle=\sum_{k=1}^mx_k{\ell}\left(\frac{x_k}{p_k}\right)$ of the parallel neural network as a whole at the step~$n$. The key idea is that the distribution of the synaptic link whose damage is smaller (larger) than the cumulative damage of the neural network will increase (decrease) over the neurotransmission process.

We now propose a novel mathematical model of distributions of toxic neuropeptides/neurotransmitters by a discrete-time replicator equation $\mathcal{R}:\mathbb{S}^{m-1}\to\mathbb{S}^{m-1}$ 
\begin{equation}\label{REdefinedbyRM}
	\left(\mathcal{R}(\mathbf{x})\right)_k=x_k\left[1+\varepsilon\left(\ell\left(\frac{x_k}{p_k}\right)-\sum\limits_{i=1}^mx_i\ell\left(\frac{x_i}{p_i}\right)\right)\right], \quad \forall \ k\in\mathbf{I}_m
\end{equation}
where $\ell:[0,\bar{p}]\to[0,+\infty)$ is a convex, continuously differentiable, and strictly increasing function, $\mathbf{p}=(p_1,\ldots,p_m)\in\textup{int}\mathbb{S}^{m-1}$ is a positive distribution, $\bar{p}:=\frac{1}{\min\limits_{i\in\mathbf{I}_m}{p_i}}$, and $\varepsilon\in(-1,0)$. 
Some particular cases of equation \eqref{REdefinedbyRM} were studied in the literature \cite{Saburov2021,Saburov2022}.

\begin{definition}[Nash Equilibrium]\label{NE}
	A distribution $\mathbf{x}\in \mathbb{S}^{m-1}$ is called a \textit{Nash equilibrium} of the replicator equation given by \eqref{REdefinedbyRM} if one has $\langle\mathbf{x},\varepsilon\mathbf{\Phi}_m(\frac{\mathbf{x}}{\mathbf{p}})\rangle\geq \langle\mathbf{y},\varepsilon\mathbf{\Phi}_m(\frac{\mathbf{x}}{\mathbf{p}})\rangle$ for any $\mathbf{y}\in\mathbb{S}^{m-1}$. A flow $\mathbf{x}$ is called a \textit{strictly Nash equilibrium} if one has $\langle\mathbf{x},\varepsilon\mathbf{\Phi}_m(\frac{\mathbf{x}}{\mathbf{p}})\rangle> \langle\mathbf{y},\varepsilon\mathbf{\Phi}_m(\frac{\mathbf{x}}{\mathbf{p}})\rangle$ for any $\mathbf{y}\in\mathbb{S}^{m-1}$ with $\mathbf{y}\neq\mathbf{x}$. 
\end{definition}

We define 
$$
\mu:=\max\limits_{z\in[0,\bar{p}]}\frac{d}{dz}\left(z\ell(z)\right), \qquad \mathbf{p}_\alpha:=\frac{1}{s_\alpha(\mathbf{p})}\sum_{i\in\alpha}p_i\mathbf{e}_i
$$ 
for all $\alpha\subset\mathbf{I}_m$,  where $s_\alpha(\mathbf{p})=\sum_{i\in\alpha}p_i$, $\mathbf{p}=(p_1,\ldots,p_m)\in\textup{int}\mathbb{S}^{m-1}$, and $\mathbf{e}_i$ is the vertex of the simplex ${\mathbb{S}}^{m-1}$ for all $i\in\mathbf{I}_m$. Denote $\textbf{Fix}(\mathcal{R})=\{\mathbf{x}\in\mathbb{S}^{m-1}: \mathcal{R}(\mathbf{x})=\mathbf{x}\}$ a set of \textit{fixed}  points of the replicator equation \eqref{REdefinedbyRM}.

The next result  is the main result of this paper that describes the dynamics of distributions of toxic neuropeptides/neurotransmitters during the neurotransmission process for sufficiently small $\varepsilon\in(-1,0)$.

\begin{atheorem}\label{DynamicsofWOF}
	Let $\varepsilon\in(-\frac{1}{\mu},0)\cap(-1,0)$. Then the following statements hold true: 
	\begin{description}\setlength\parskip{1pt}
		\item[$(i)$] One has $\textbf{\textup{Fix}}(\mathcal{R})=\bigcup\limits_{\alpha\subset\mathbf{I}_m}\{\mathbf{p}_\alpha\};$ 
		\item[$(ii)$] The unique neurotransmitters-synapses optimum $\mathbf{p}$ is the only Nash equilibrium;
		\item[$(iii)$] The unique neurotransmitters-synapses optimum $\mathbf{p}$ is the only stable fixed point; 
		\item[$(iv)$] An interior orbit converges to the unique neurotransmitters-synapses optimum $\mathbf{p}$.
	\end{description}
\end{atheorem} 

\subsection{Lyapunov Functions}

We study the dynamics  of the discrete-time replicator equation \eqref{REdefinedbyRM} by means of a Lyapunov function. Namely, in order to study the stability of fixed points of the replicator equation, we employ a Lyapunov function.

\begin{definition}[Lyapunov Function]\label{LF} A continuous function $\varphi:\mathbb{S}^{m-1}\to\mathbb{R}$ is called a \textit{Lyapunov function} if the number sequence $$\{\varphi(\mathbf{x}), \varphi(\mathcal{R}(\mathbf{x})), \cdots, \varphi(\mathcal{R}^{(n)}(\mathbf{x})), \cdots \}$$ is a bounded monotone sequence for any initial point $\mathbf{x}\in\mathbb{S}^{m-1}$. 
\end{definition}

\begin{proposition}\label{Lyapunovfunctionnegative} Let $\varepsilon\in(-\frac{1}{\mu},0)\cap(-1,0)$.
	Then the following statements hold true:
	\begin{description}
		\item[$(i)$]  A continuous function $\mathcal{M}_{\mathbf{p}:{k}}(\mathbf{x}):=\max\limits_{i\in\mathbf{I}_m}\{\frac{x_{i}}{p_i}\}-\frac{x_k}{p_k}$ for all $k\in\mathbf{I}_m$ ia a decreasing Lyapunov function for $\mathcal{R}:\textup{int}\mathbb{S}^{m-1}\to
		\textup{int}\mathbb{S}^{m-1}$;
		\item[$(ii)$] A continuous function $\mathcal{M}_{\mathbf{p}_\alpha: k}(\mathbf{x}):=\max\limits_{i\in\alpha}\{\frac{x_{i}}{p_i}\}-\frac{x_k}{p_k}$ for all  $k\in\alpha\subset\mathbf{I}_m$ is a decreasing Lyapunov function for $\mathcal{R}:\textup{int}\mathbb{S}^{|\alpha|-1}\to
		\textup{int}\mathbb{S}^{|\alpha|-1}$.
	\end{description}
\end{proposition}

\begin{proof}
	Here, we only present the proof of the part $(i)$. The part $(ii)$ can be similarly proved. Its proof is omitted here. 
	
	We show that $\mathcal{M}_{\mathbf{p}:{k}}(\mathbf{x}):=\max\limits_{i\in\mathbf{I}_m}\{\frac{x_{i}}{p_i}\}-\frac{x_k}{p_k}$ for all $k\in\mathbf{I}_m$ is a decreasing Lyapunov function along an orbit (trajectory) of the operator $\mathcal{R}:\textup{int}\mathbb{S}^{m-1}\to
	\textup{int}\mathbb{S}^{m-1}$. 
	
	Let $\mathbf{x}\in\textup{int}\mathbb{S}^{m-1}$. It follows from \eqref{REdefinedbyRM} for any $k,t\in\mathbf{I}_m$ that if $x_t=x_k$ then $(\mathcal{R}(\mathbf{x}))_t=(\mathcal{R}(\mathbf{x}))_k$ and  if $x_t\neq x_k$ then 
	
	\begin{equation*}
		\frac{(\mathcal{R}(\mathbf{x}))_t}{p_t}-\frac{(\mathcal{R}(\mathbf{x}))_k}{p_k}
		=\left(\frac{x_t}{p_t}-\frac{x_k}{p_k}\right)\left[1+\varepsilon\left(\frac{\frac{x_t}{p_t}\ell(\frac{x_t}{p_t})-\frac{x_k}{p_k}\ell(\frac{x_k}{p_k})}{\frac{x_t}{p_t}-\frac{x_k}{p_k}}-\sum\limits_{i=1}^mx_i\ell\left(\frac{x_i}{p_i}\right)\right)\right].
	\end{equation*}
	
	Since $-\frac{1}{\mu}<\varepsilon<0$ and for all $t,k\in\mathbf{I}_m$
	$$
	\left|\frac{\frac{x_t}{p_t}\ell(\frac{x_t}{p_t})-\frac{x_k}{p_k}\ell(\frac{x_k}{p_k})}{\frac{x_t}{p_t}-\frac{x_k}{p_k}}\right|\leq \max\limits_{z\in[0,\bar{p}]}\frac{d}{dz}\left(z\ell(z)\right)=\mu,
	$$
	we obtain  
	$$\mathsf{Sign}\left(\frac{(\mathcal{R}(\mathbf{x}))_t}{p_t}-\frac{(\mathcal{R}(\mathbf{x}))_k}{p_k}\right)=\mathsf{Sign}\left(\frac{x_t}{p_t}-\frac{x_k}{p_k}\right)$$ for any $k,t\in\mathbf{I}_m$ and $$\mathsf{MaxInd}\left(\frac{\mathcal{R}(\mathbf{x})}{\mathbf{p}}\right)=\mathsf{MaxInd}\left(\frac{\mathbf{x}}{\mathbf{p}}\right), \quad \forall \ \mathbf{x}\in\textup{int}\mathbb{S}^{m-1}$$ 
	where  $\frac{\mathbf{u}}{\mathbf{v}}:=\left(\frac{u_1}{v_1},\cdots, \frac{u_m}{v_m}\right)$ for any two vectors $\mathbf{u}=(u_1,\cdots,u_m)$ and $\mathbf{v}=(v_1,\cdots,v_m)>0$ and $\mathsf{MaxInd}\left(\frac{\mathbf{u}}{\mathbf{v}}\right):=\{k\in\mathbf{I}_m: \frac{u_k}{v_k}=\max\limits_{i\in\mathbf{I}_m}\{\frac{u_i}{v_i}\}\}$.
	
	Moreover, if $t\in\mathsf{MaxInd}\left(\frac{\mathcal{R}(\mathbf{x})}{\mathbf{p}}\right)=\mathsf{MaxInd}\left(\frac{\mathbf{x}}{\mathbf{p}}\right)$ then for all $k\in\mathbf{I}_m$ we get
	\begin{multline*}
		\frac{(\mathcal{R}(\mathbf{x}))_t}{p_t}-\frac{(\mathcal{R}(\mathbf{x}))_k}{p_k}\\
		=\left(\frac{x_t}{p_t}-\frac{x_k}{p_k}\right)
		+\varepsilon\left(\frac{x_t}{p_t}-\frac{x_k}{p_k}\right)\left(\frac{\frac{x_k}{p_k}\left(\ell(\frac{x_t}{p_t})-\ell(\frac{x_k}{p_k})\right)}{\frac{x_t}{p_t}-\frac{x_k}{p_k}}+\sum\limits_{i=1}^mx_i\left(\ell(\frac{x_t}{p_t})-\ell(\frac{x_i}{p_i})\right)\right).
	\end{multline*}
	Since $-1<\varepsilon<0$, we obtain that $\mathcal{M}_{\mathbf{p}:{k}}\left(\mathcal{R}(\mathbf{x})\right)\leq \mathcal{M}_{\mathbf{p}:{k}}(\mathbf{x})$ for all $k\in\mathbf{I}_m$.
	
	By repeating this process, we get for all $k\in\mathbf{I}_m$, $n\in\mathbb{N}$, and $\mathbf{x}\in\textup{int}\mathbb{S}^{m-1}$ that
	$$
	\mathsf{MaxInd}\left(\frac{\mathcal{R}^{(n)}(\mathbf{x})}{\mathbf{p}}\right)=\mathsf{MaxInd}\left(\frac{\mathbf{x}}{\mathbf{p}}\right), \quad \mathcal{M}_{\mathbf{p}:{k}}\left(\mathcal{R}^{(n+1)}(\mathbf{x})\right)\leq \mathcal{M}_{\mathbf{p}:{k}}\left(\mathcal{R}^{(n)}(\mathbf{x})\right).
	$$
	
	This shows that $\mathcal{M}_{\mathbf{p}:{k}}(\mathbf{x})$ for $k\in\mathbf{I}_m$ is a Lyapunov function over the set $\textup{int}\mathbb{S}^{m-1}$. This completes the proof.
\end{proof}

\subsection{Asymptotic Stability}

We can describe the asymptotic stability of the discrete-time replicator equation \eqref{REdefinedbyRM} all over the simplex $\mathbb{S}^{m-1}$. 

\begin{definition}[Stable and Attracting Points]
	A fixed point $\mathbf{y}\in\mathbb{S}^{m-1}$ is called \textit{stable} if for every neighborhood $U(\mathbf{y})\subset\mathbb{S}^{m-1}$ of $\mathbf{y}$ there exists a neighborhood $V(\mathbf{y})\subset\mathbb{S}^{m-1}$ of $\mathbf{y}$ such that an orbit $\{\mathbf{x}, \mathcal{R}(\mathbf{x}), \cdots, \mathcal{R}^{(n)}(\mathbf{x}), \cdots\}$ of any initial point $\mathbf{x}\in V(\mathbf{y})$ remains inside of the neighborhood $U(\mathbf{y})$. A fixed point $\mathbf{y}\in\mathbb{S}^{m-1}$ is called \textit{attracting} if there exists a neighborhood $V(\mathbf{y})\subset\mathbb{S}^{m-1}$ of $\mathbf{y}$ such that an orbit $\{\mathbf{x}, \mathcal{R}(\mathbf{x}), \cdots, \mathcal{R}^{(n)}(\mathbf{x}), \cdots\}$ of any initial point $\mathbf{x}\in V(\mathbf{y})$ converges to $\mathbf{y}$. A fixed point $\mathbf{y}\in\mathbb{S}^{m-1}$ is called \textit{asymptotically stable} if it is both stable and attracting. 
\end{definition}

\begin{proposition}\label{Dynamicsfornegative} 
	Let $\varepsilon\in(-\frac{1}{\mu},0)\cap(-1,0)$. Then an orbit of the replicator equation $\mathcal{R}:\mathbb{S}^{m-1}\to\mathbb{S}^{m-1}$ starting from any initial point $\mathbf{x}\in\mathbb{S}^{m-1}$ converges to the fixed point $\mathbf{p}_{\textup{supp}(\mathbf{x})}$ in the interior of the face $\mathbb{S}^{|\textup{supp}(\mathbf{x})|-1}$.
\end{proposition}

\begin{proof}
	Without loss of generality, we may assume that $\mathbf{x}\in\textup{int}\mathbb{S}^{m-1}$. Otherwise, we choose a suitable Lyapunov function given in the part $(ii)$ of Proposition~\ref{Lyapunovfunctionnegative} depending on an initial point $\mathbf{x}\in\mathbb{S}^{m-1}$.
	
	We fix $t_0\in\mathsf{MaxInd}\left(\frac{\mathbf{x}}{\mathbf{p}}\right)$. As shown in Proposition \ref{Lyapunovfunctionnegative}, we obtain for any $n\in\mathbb{N}$ that  $$\mathsf{MaxInd}\left(\frac{\mathcal{R}^{(n)}(\mathbf{x})}{\mathbf{p}}\right)=\mathsf{MaxInd}\left(\frac{\mathbf{x}}{\mathbf{p}}\right).$$ 
	
	Since the sequence $
	\left\{\mathcal{M}_{\mathbf{p}:{k}}\left(\mathcal{R}^{(n)}(\mathbf{x})\right)\right\}_{n=0}^\infty$ where
	$$
	\mathcal{M}_{\mathbf{p}:{k}}\left(\mathcal{R}^{(n)}(\mathbf{x})\right)=\left(\mathcal{R}^{(n)}(\mathbf{x})\right)_{t_0}-\left(\mathcal{R}^{(n)}(\mathbf{x})\right)_k
	$$
	is convergent for each $k\in\mathbf{I}_m$, the sequence $\left\{\left(\mathcal{R}^{(n)}(\mathbf{x})\right)_{t_0}\right\}_{n=0}^\infty$ where
	$$
	\left(\mathcal{R}^{(n)}(\mathbf{x})\right)_{t_0}=\frac{1}{m}\left(1+\sum_{k=1}^m\mathcal{M}_{\mathbf{p}:{k}}\left(\mathcal{R}^{(n)}(\mathbf{x})\right)\right)
	$$
	is also convergent ($m$ is the number of total links $\mathbf{I}_m=\{1,2,\cdots m\}$). 
	
	Hence, the sequence $\left\{\left(\mathcal{R}^{(n)}(\mathbf{x})\right)_{k}\right\}_{n=0}^\infty$ for each $k\in\mathbf{I}_m$ where
	$$
	\left(\mathcal{R}^{(n)}(\mathbf{x})\right)_k=\left(\mathcal{R}^{(n)}(\mathbf{x})\right)_{t_0}-\mathcal{M}_{\mathbf{p}:{k}}\left(\mathcal{R}^{(n)}(\mathbf{x})\right)
	$$
	is also convergent. This means that the orbit (trajectory) $\left\{\mathcal{R}^{(n)}(\mathbf{x})\right\}_{n=0}^\infty$ of the discrete-time replicator equation  $\mathcal{R}:\textup{int}\mathbb{S}^{m-1}\to
	\textup{int}\mathbb{S}^{m-1}$ starting from any initial point $\mathbf{x}\in\textup{int}\mathbb{S}^{m-1}$ is convergent and its omega limiting point is some fixed point $\mathbf{x}^{*}$. 
	
	Moreover, we now show that the orbit (trajectory) $\left\{\mathcal{R}^{(n)}(\mathbf{x})\right\}_{n=0}^\infty$ is separated from the boundary $\partial\mathbb{S}^{m-1}$ of the simplex $\mathbb{S}^{m-1}$, i.e., for some $\delta_0>0$ we have 
	$$
	\min\limits_{k\in\mathbf{I}_m}\frac{\left(\mathcal{R}^{(n)}(\mathbf{x})\right)_k}{p_k}\geq \delta_0, \quad \forall \ n\in\mathbb{N}.
	$$
	
	Indeed, as we already showed in the part $(i)$ of Proposition~\ref{Lyapunovfunctionnegative} that 	$$\mathsf{Sign}\left(\frac{(\mathcal{R}(\mathbf{x}))_t}{p_t}-\frac{(\mathcal{R}(\mathbf{x}))_k}{p_k}\right)=\mathsf{Sign}\left(\frac{x_t}{p_t}-\frac{x_k}{p_k}\right)$$ for any $k,t\in\mathbf{I}_m$, we obtain $$\mathsf{MinInd}\left(\frac{\mathcal{R}(\mathbf{x})}{\mathbf{p}}\right)=\mathsf{MinInd}\left(\frac{\mathbf{x}}{\mathbf{p}}\right), \quad \forall \ \mathbf{x}\in\textup{int}\mathbb{S}^{m-1}$$ 
	
	Since $-1<\varepsilon<0$, it follows from \eqref{REdefinedbyRM} and  $k\in\mathsf{MinInd}\left(\frac{\mathbf{x}}{\mathbf{p}}\right)$ that
	$$
	\frac{\left(\mathcal{R}(\mathbf{x})\right)_k}{p_k}=\frac{x_k}{p_k}\left[1-\varepsilon\sum\limits_{i=1}^mx_i\left(\ell\left(\frac{x_i}{p_i}\right)-\ell\left(\frac{x_k}{p_k}\right)\right)\right]\geq \frac{x_k}{p_k}.
	$$
	
	This means that $	\min\limits_{k\in\mathbf{I}_m}\frac{\left(\mathcal{R}(\mathbf{x})\right)_k}{p_k}\geq\delta_0$ for $\delta_0:=\min\limits_{k\in\mathbf{I}_m}\frac{x_k}{p_k}>0$ (since $\mathbf{x}\in\textup{int}\mathbb{S}^{m-1}$).
	
	By repeating this process, we obtain that
	$$
	\mathsf{MinInd}\left(\frac{\mathcal{R}^{(n)}(\mathbf{x})}{\mathbf{p}}\right)=\mathsf{MinInd}\left(\frac{\mathbf{x}}{\mathbf{p}}\right),\quad \min\limits_{k\in\mathbf{I}_m}\frac{\left(\mathcal{R}^{(n)}(\mathbf{x})\right)_k}{p_k}\geq \delta_0, \quad \forall \ n\in\mathbb{N}.
	$$ 
	
	Since the trajectory $\left\{\mathcal{R}^{(n)}(\mathbf{x})\right\}_{n=0}^\infty$ of the operator  $\mathcal{R}:\textup{int}\mathbb{S}^{m-1}\to
	\textup{int}\mathbb{S}^{m-1}$ is separated from the boundary $\partial\mathbb{S}^{m-1}$ of the simplex $\mathbb{S}^{m-1}$, its limiting point $\mathbf{x}^{*}$ belongs to the interior $\textup{int}\mathbb{S}^{m-1}$ of the simplex $\mathbb{S}^{m-1}$. Then the only fixed point in the interior $\textup{int}\mathbb{S}^{m-1}$ of the simplex $\mathbb{S}^{m-1}$ is the Wardrop optimal flow $\mathbf{p}$. Hence, the omega limiting point of the trajectory starting from any initial point $\mathbf{x}\in\textup{int}\mathbb{S}^{m-1}$ is the Wardrop optimal flow $\mathbf{p}\in\textup{int}\mathbb{S}^{m-1}$.
	This completes the proof.
\end{proof}

We need the following result in the proof of the main result Theorem~\ref{DynamicsofWOF}.

\begin{proposition}\label{simplefact}
	Let $\mathbf{x},\mathbf{y}\in\mathbb{S}^{m-1}$ be two {\textbf{distinct}} elements of the simplex $\mathbb{S}^{m-1}$ such that $\mathbf{y}>0$. Let $\frac{\mathbf{x}}{\mathbf{y}}:=(\frac{x_1}{y_1},\cdots,\frac{x_m}{y_m})$, $\bar{y}:=\min\limits_{k\in\mathbf{I}_m}y_k>0$ and $\|\mathbf{x}-\mathbf{y}\|_1:=\sum_{k\in\mathbf{I}_m}|x_k-y_k|$. Then the following statements hold true:
	\begin{description}
		\item[$(i)$] $\min\limits_{k\in\mathbf{I}_m}\frac{x_k}{y_k}<1<\max\limits_{k\in\mathbf{I}_m}\frac{x_k}{y_k}$
		\item[$(ii)$] $\bar{y}\left[\max\limits_{k\in\mathbf{I}_m}\frac{x_k}{y_k}-\min\limits_{k\in\mathbf{I}_m}\frac{x_k}{y_k}\right]\leq \|\mathbf{x}-\mathbf{y}\|_1 \leq m\left[\max\limits_{k\in\mathbf{I}_m}\frac{x_k}{y_k}-\min\limits_{k\in\mathbf{I}_m}\frac{x_k}{y_k}\right]$
	\end{description}
\end{proposition}

\begin{proof} Since $\mathbf{x}\neq\mathbf{y}$ and $\mathbf{x},\mathbf{y}\in\mathbb{S}^{m-1}$, we have that $x_{k_0}\neq y_{k_0}$ for some $k_0\in\mathbf{I}_m$ and 
	$$\sum_{k\in\mathbf{I}_m}x_k=\sum_{k\in\mathbf{I}_m}y_k=1.$$
	
	$(i)$ Let us first assume that $\max\limits_{k\in\mathbf{I}_m}\frac{x_k}{y_k}\leq 1$. Then $\frac{x_k}{y_k}\leq 1$ for all  $k\in\mathbf{I}_m$ or equivalently ${x_k}\leq {y_k}$ for all $k\in\mathbf{I}_m$. Since $x_{k_0}\neq y_{k_0}$ for some $k_0\in\mathbf{I}_m$, we have that ${x_{k_0}}<{y_{k_0}}$. Consequently, we have
	$$
	1=\sum_{k\in\mathbf{I}_m}x_k<\sum_{k\in\mathbf{I}_m}y_k=1
	$$
	
	However, this is a contradiction. So, we must have $\max\limits_{k\in\mathbf{I}_m}\frac{x_k}{y_k}>1$.
	
	Let us now assume that $\min\limits_{k\in\mathbf{I}_m}\frac{x_k}{y_k}\geq 1$. Then $\frac{x_k}{y_k}\geq 1$ for all  $k\in\mathbf{I}_m$ or equivalently ${x_k}\geq {y_k}$ for all $k\in\mathbf{I}_m$. Since $x_{k_0}\neq y_{k_0}$ for some $k_0\in\mathbf{I}_m$, we have that ${x_{k_0}}>{y_{k_0}}$. Consequently, we have
	$$
	1=\sum_{k\in\mathbf{I}_m}x_k>\sum_{k\in\mathbf{I}_m}y_k=1
	$$
	
	However, this is a contradiction. So, we must have $\min\limits_{k\in\mathbf{I}_m}\frac{x_k}{y_k}<1$.
	
	Hence, we always have $\min\limits_{k\in\mathbf{I}_m}\frac{x_k}{y_k}<1<\max\limits_{k\in\mathbf{I}_m}\frac{x_k}{y_k}$.
	
	$(ii)$ On the one hand, we have 
	\begin{multline*}
		\|\mathbf{x}-\mathbf{y}\|_1:=\sum_{k\in\mathbf{I}_m}|x_k-y_k|=\sum_{k\in\mathbf{I}_m}y_k\left|\frac{x_k}{y_k}-1\right|\geq \sum_{k\in\mathbf{I}_m}\bar{y}\left|\frac{x_k}{y_k}-1\right|\\
		\geq \bar{y}\left|\max\limits_{k\in\mathbf{I}_m}\frac{x_k}{y_k}-1\right|+\bar{y}\left|\min\limits_{k\in\mathbf{I}_m}\frac{x_k}{y_k}-1\right|=\bar{y}\left[\left(\max\limits_{k\in\mathbf{I}_m}\frac{x_k}{y_k}-1\right)+\left(1-\min\limits_{k\in\mathbf{I}_m}\frac{x_k}{y_k}\right)\right]\\
		=\bar{y}\left[\max\limits_{k\in\mathbf{I}_m}\frac{x_k}{y_k}-\min\limits_{k\in\mathbf{I}_m}\frac{x_k}{y_k}\right]
	\end{multline*}
	
	On the other hand, since $1$ and $\frac{x_k}{y_k}$ belong to the segment $\left[\min\limits_{k\in\mathbf{I}_m}\frac{x_k}{y_k}, \ \max\limits_{k\in\mathbf{I}_m}\frac{x_k}{y_k}\right]$ we have 
	$$
	\left|\frac{x_k}{y_k}-1\right|\leq \left[\max\limits_{k\in\mathbf{I}_m}\frac{x_k}{y_k}-\min\limits_{k\in\mathbf{I}_m}\frac{x_k}{y_k}\right],  \quad \forall \ k\in\mathbf{I}_m.
	$$  
	Consequently, by means of $0<y_k<1$, we derive 
	\begin{multline*}
		\|\mathbf{x}-\mathbf{y}\|_1:=\sum_{k\in\mathbf{I}_m}|x_k-y_k|=\sum_{k\in\mathbf{I}_m}y_k\left|\frac{x_k}{y_k}-1\right|\leq \sum_{k\in\mathbf{I}_m}\left|\frac{x_k}{y_k}-1\right|\\
		\leq \sum_{k\in\mathbf{I}_m}\left[\max\limits_{k\in\mathbf{I}_m}\frac{x_k}{y_k}-\min\limits_{k\in\mathbf{I}_m}\frac{x_k}{y_k}\right]=m\left[\max\limits_{k\in\mathbf{I}_m}\frac{x_k}{y_k}-\min\limits_{k\in\mathbf{I}_m}\frac{x_k}{y_k}\right]
	\end{multline*}
	This completes the proof.
\end{proof}

\subsection{The Proof of Theorem \ref{DynamicsofWOF}}

We now provide the proof of Theorem \ref{DynamicsofWOF}. 

$(i)$
We first show  $\textbf{Fix}(\mathcal{R})=\bigcup\limits_{\alpha\subset\mathbf{I}_m}\{\mathbf{p}_\alpha\}$. It is obvious that $\bigcup\limits_{\alpha\subset\mathbf{I}_m}\{\mathbf{p}_\alpha\}\subset\textbf{Fix}(\mathcal{R})$. Let $\mathbf{x}\in\textbf{Fix}(\mathcal{R})$ be any fixed point. We set $\alpha:=\textup{supp}(\mathbf{x})$. It follows from \eqref{REdefinedbyRM} that
$$
\ell\left(\frac{x_k}{p_k}\right)=\sum\limits_{i\in\alpha}x_i\ell\left(\frac{x_i}{p_i}\right), \quad \forall \ k\in\alpha.
$$
This means that $\ell\left(\frac{x_{k_1}}{p_{k_1}}\right)=\ell\left(\frac{x_{k_2}}{p_{k_2}}\right)$	for any distinct $k_1,k_2\in\alpha$. Since $\ell$ is strictly increasing, we obtain $\frac{x_{k_1}}{p_{k_1}}=\frac{x_{k_2}}{p_{k_2}}$ for any $k_1,k_2\in\alpha$. Hence, we get $\mathbf{x}=\mathbf{p}_\alpha$. This shows $\textbf{Fix}(\mathcal{R})=\bigcup\limits_{\alpha\subset\mathbf{I}_m}\{\mathbf{p}_\alpha\}$.

$(ii)$ We now describe the set $\textbf{NE}(\mathcal{R})$ of all Nash equilibria of the discrete-time replicator equation. Namely, we now show $\textbf{NE}(\mathcal{R})=\{\mathbf{p}\}$. For the sake of convenience, we set  $\mathcal{E}_{\varepsilon\mathbf{\Phi}_m}(\mathbf{y},\mathbf{x}):=\langle\mathbf{y},\varepsilon\mathbf{\Phi}_m(\frac{\mathbf{x}}{\mathbf{p}})\rangle$ for all $\mathbf{x},\mathbf{y}\in\mathbb{S}^{m-1}$. For any $\mathbf{y}\in\mathbb{S}^{m-1}$, it is obvious that 
$$
\mathcal{E}_{\varepsilon\mathbf{\Phi}_m}(\mathbf{p},\mathbf{p})=\varepsilon\sum\limits_{i\in\mathbf{I}_m}p_i\ell\left(\frac{p_i}{p_i}\right)=\varepsilon \ell(1)= \varepsilon\sum\limits_{i\in\mathbf{I}_m}y_i\ell\left(\frac{p_i}{p_i}\right)=\mathcal{E}_{\varepsilon\mathbf{\Phi}_m}(\mathbf{y},\mathbf{p}),
$$
which means that $\mathbf{p}\in\textbf{NE}(\mathcal{R})$.

Let $\mathbf{x}\in\textbf{NE}(\mathcal{R})$ be any Nash equilibrium where $\alpha:=\textup{supp}(\mathbf{x})$. We then want to prove that $\alpha=\mathbf{I}_m$. We assume the contrary, i.e., $\alpha\neq \mathbf{I}_m$. Let $k\in\mathbf{I}_m\setminus\alpha$ (since $\mathbf{I}_m\setminus\alpha\neq\emptyset$). Since $-1<\varepsilon<0$  and $\mathbf{x}\in\textbf{NE}(\mathcal{R})$, we face the contradiction
$$
0=\mathcal{E}_{\varepsilon\mathbf{\Phi}_m}(\mathbf{e}_k,\mathbf{x})\leq\mathcal{E}_{\varepsilon\mathbf{\Phi}_m}(\mathbf{x},\mathbf{x})=\varepsilon\sum_{i\in\alpha} x_i\ell\left(\frac{x_i}{p_i}\right)<0.
$$

This shows $\alpha=\mathbf{I}_m$. In this case, we have for any $\mathbf{y}\in\mathbb{S}^{m-1}$ 
$$
\varepsilon \max\limits_{i\in\mathbf{I}_m}\ell\left(\frac{x_i}{p_i}\right)\leq\mathcal{E}_{\varepsilon\mathbf{\Phi}_m}(\mathbf{y},\mathbf{x})=\varepsilon\sum_{i\in\mathbf{I}_m} y_i\ell\left(\frac{x_i}{p_i}\right)\leq \varepsilon \min\limits_{i\in\mathbf{I}_m}\ell\left(\frac{x_i}{p_i}\right). 
$$	

Since $\mathbf{x}\in\textbf{NE}(\mathcal{R})$, we must have $$\mathcal{E}_{\varepsilon\mathbf{\Phi}_m}(\mathbf{x},\mathbf{x})=\max\limits_{\mathbf{y}\in\mathbb{S}^{m-1}}\mathcal{E}_{\varepsilon\mathbf{\Phi}_m}(\mathbf{y},\mathbf{x})=\varepsilon\min\limits_{i\in\mathbf{I}_m}\ell\left(\frac{x_i}{p_i}\right).$$

On the other hand, we have 
$$
\mathcal{E}_{\varepsilon\mathbf{\Phi}_m}(\mathbf{x},\mathbf{x})=\varepsilon\sum_{i\in\mathbf{I}_m} x_i\ell\left(\frac{x_i}{p_i}\right)\leq\varepsilon\min\limits_{i\in\mathbf{I}_m}\ell\left(\frac{x_i}{p_i}\right). 
$$

Consequently, we obtain  $\ell\left(\frac{x_k}{p_k}\right)=\min\limits_{i\in\mathbf{I}_m}\ell\left(\frac{x_i}{p_i}\right)$ for any $k\in\mathbf{I}_m$. Since $\ell$ is strictly increasing, we obtain that $\frac{x_{k_1}}{p_{k_1}}=\frac{x_{k_2}}{p_{k_2}}$ for any $k_1,k_2\in\mathbf{I}_m$ or equivalently $\mathbf{x}=\mathbf{p}$. This shows $\textbf{NE}(\mathcal{R})=\{\mathbf{p}\}$, i.e., the unique Wardrop optimal flow $\mathbf{p}$ is the only Nash equilibrium.

$(iii)$ We now show that the only stable fixed point is the Wardrop optimal flow $\mathbf{p}=(p_1,\cdots,p_m)$ which is also the only Nash equilibrium. 

On the one hand, according to Proposition~\ref{Dynamicsfornegative}, if  $|\alpha|<m$ for $\alpha\subset\mathbf{I}_m$ (which means $\alpha\neq\mathbf{I}_m$) then the fixed point $\mathbf{p}_\alpha$ of the face $\mathbb{S}^{|\alpha|-1}$ (which means different from $\mathbf{p}$) is not stable. Indeed, for any small neighborhood $U(\mathbf{p}_\alpha)\subset \mathbb{S}^{m-1}$ of the fixed point $\mathbf{p}_\alpha$ there is a point $\mathbf{x}\in U(\mathbf{p}_\alpha)\cap\textup{int}\mathbb{S}^{m-1}$  such that an orbit of the interior point $\mathbf{x}$ converges to the Wardrop optimal flow $\mathbf{p}$. Therefore, the fixed point $\mathbf{p}_\alpha$ of the face $\mathbb{S}^{|\alpha|-1}$ is not stable whenever $|\alpha|<m$ for any $\alpha\subset\mathbf{I}_m$  (which means $\alpha\neq\mathbf{I}_m$). 

On the other hand, according to Proposition~\ref{Dynamicsfornegative}, the Wardrop optimal flow $\mathbf{p}\in\textup{int}\mathbb{S}^{m-1}$ is stable.  Namely, for any ball $U_R(\mathbf{p})\Subset \textup{int}\mathbb{S}^{m-1}$ of radius $R$ there exists a (small) ball $U_r(\mathbf{p})\Subset \textup{int}\mathbb{S}^{m-1}$ of radius $r$ (we may choose $r=\frac{\bar{p}R}{m}$) such that $\|\mathcal{R}^{(n)}(\mathbf{x})-\mathbf{p}\|_1\leq R$ for any $n\in\mathbb{R}$ whenever $\|\mathbf{x}-\mathbf{p}\|_1\leq r$. 

Indeed, it follows from Proposition \ref{simplefact}  
$$
\bar{p}\left[\max\limits_{i\in\mathbf{I}_m}\frac{x_i}{p_i}-\min\limits_{i\in\mathbf{I}_m}\frac{x_i}{p_i}\right]\leq\|\mathbf{x}-\mathbf{p}\|_1\leq r
$$ 

Since $-1<\varepsilon<0$, for $k_1\in\mathsf{MinInd}\left(\frac{\mathbf{x}}{\mathbf{p}}\right)$ and $k_2\in\mathsf{MaxInd}\left(\frac{\mathbf{x}}{\mathbf{p}}\right)$ it follows from \eqref{REdefinedbyRM} that
\begin{eqnarray*}
	\frac{\left(\mathcal{R}(\mathbf{x})\right)_{k_1}}{p_{k_1}}=\frac{x_{k_1}}{p_{k_1}}\left[1-\varepsilon\sum\limits_{i=1}^mx_i\left(\ell\left(\frac{x_{i}}{p_{i}}\right)-\ell\left(\frac{x_{k_1}}{p_{k_1}}\right)\right)\right]\geq \frac{x_{k_1}}{p_{k_1}},\\
	\frac{\left(\mathcal{R}(\mathbf{x})\right)_{k_2}}{p_{k_2}}=\frac{x_{k_2}}{p_{k_2}}\left[1+\varepsilon\sum\limits_{i=1}^mx_i\left(\ell\left(\frac{x_{k_2}}{p_{k_2}}\right)-\ell\left(\frac{x_{i}}{p_{i}}\right)\right)\right]\leq \frac{x_{k_2}}{p_{k_2}}.
\end{eqnarray*}
Since 
$$
\mathsf{MinInd}\left(\frac{\mathcal{R}^{(n)}(\mathbf{x})}{\mathbf{p}}\right)=\mathsf{MinInd}\left(\frac{\mathbf{x}}{\mathbf{p}}\right) \quad  \textup{and} \quad	\mathsf{MaxInd}\left(\frac{\mathcal{R}^{(n)}(\mathbf{x})}{\mathbf{p}}\right)=\mathsf{MaxInd}\left(\frac{\mathbf{x}}{\mathbf{p}}\right),$$
by repeating this process, we obtain
\begin{eqnarray*}
	\min_{i\in\mathbf{I}_m}\frac{x_i}{p_i}\leq \cdots\leq\min_{i\in\mathbf{I}_m}\frac{(\mathcal{R}^{(n)}(\mathbf{x}))_i}{p_i}\leq \cdots \leq \max_{i\in\mathbf{I}_m}\frac{(\mathcal{R}^{(n)}(\mathbf{x}))_i}{p_i} \leq \cdots\leq \max_{i\in\mathbf{I}_m}\frac{x_i}{p_i}
\end{eqnarray*}
Hence, we have
\begin{eqnarray*}
	\left[\max_{i\in\mathbf{I}_m}\frac{(\mathcal{R}^{(n)}(\mathbf{x}))_i}{p_i}-\min_{i\in\mathbf{I}_m}\frac{(\mathcal{R}^{(n)}(\mathbf{x}))_i}{p_i}\right]\leq \cdots \leq \left[\max_{i\in\mathbf{I}_m}\frac{x_i}{p_i}-\min_{i\in\mathbf{I}_m}\frac{x_i}{p_i}\right]\leq \frac{r}{\bar{p}}.
\end{eqnarray*}

Consequently, it follows from Proposition \ref{simplefact} 
\begin{eqnarray*}
	\|\mathcal{R}^{(n)}(\mathbf{x})-\mathbf{p}\|_1\leq m \left[\max_{i\in\mathbf{I}_m}\frac{(\mathcal{R}^{(n)}(\mathbf{x}))_i}{p_i}-\min_{i\in\mathbf{I}_m}\frac{(\mathcal{R}^{(n)}(\mathbf{x}))_i}{p_i}\right]\leq m\cdot \frac{r}{\bar{p}}=R.
\end{eqnarray*}
for any $n\in\mathbb{R}$ whenever $\|\mathbf{x}-\mathbf{p}\|_1\leq r$. This means that the Wardrop optimal flow $\mathbf{p}=(p_1,\cdots,p_m)$ is the stable fixed point. 

$(iv)$ Due to Proposition~\ref{Dynamicsfornegative}, any interior orbit starting from any interior initial point $\mathbf{x}\in \textup{int}\mathbb{S}^{m-1}$ converges to the Wardrop optimal flow $\mathbf{p}\in\textup{int}\mathbb{S}^{m-1}$. This completes the proof of Theorem~\ref{DynamicsofWOF}.

\section{Conclusion}

A part of the brain neural network that relies on chemical synapses for communication can be regarded as a \textit{transportation network}. The primary objective of this paper was to gain insight into specific mechanisms and phenomena within the neural network of the brain by applying methods originally designed for solving traffic assignment problems in transportation networks. 
In this regard, the \textit{impairment}, \textit{deterioration}, \textit{degradation}, or \textit{decline} of the functionality of synapses, occurring due to elevated levels of toxic neuropeptides/neurotransmitters during the neurotransmission process, can be linked to a \textit{traffic congestion} problem in a transportation network. Therefore, when investigating the dynamic relationship between synaptic damage and the concentration of toxic neuropeptides/neurotransmitters during the neurotransmission process, it is plausible to adapt ``\textit{Wardrop's first and second principles}''. Although these principles were originally formulated to describe \textit{optimal flow} distributions in transportation networks, we believed that they could be applied effectively within the framework of the neural network in the brain. This was a novel approach adopted in this paper.

For the sake of simplicity, our primary focus was on the investigation of the signal transmission process between two neurons through chemical synapses. The connection between two neurons, established through unidirectional synaptic links, was visualized as \textit{a parallel network between two nodes}. Within the framework of the neurotransmission process, increased levels of toxic neuropeptides/neurotransmitters damage the synaptic links. To investigate this correlation, we made an assumption that each synaptic link is associated with a \textit{synaptic damage function}. This synaptic damage function depends on a concentration of toxic neuropeptides/neurotransmitters on a synaptic link that measures its damage. In this context, we presented a complete manifestation of ``\textit{Wardrop's first and second principles}'' within a neural network of the brain. This was achieved by introducing two novel concepts: \textit{neuropeptide's (neurotransmitter's) equilibrium} and \textit{synapses optimum}. By definition, the \textit{neuropeptide/neurotransmitter equilibrium} refers to \textit{a distribution of one-unit neuropeptides/neurotransmitters that leads to uniform damage across all synaptic links}. Obviously, a such kind of distributions that leads to uniform damage across all synaptic links does not generally minimize the cumulative damage experienced by all synapses. The primary objective of experts in the domains of synapses and neurotransmitters of Alzheimer's disease is \textit{to reduce and minimize  the cumulative damage experienced by all synapses in the brains of individuals afflicted with Alzheimer's disease}. Therefore, a distribution of one-unit neuropeptides/neurotransmitters is called the \textit{optimal state of synapses} (\textit{synapses optimum}) if it minimizes the cumulative damage experienced by all synapses.  Consequently, within the context of a neural network of the brain, an analogous concept  to \textit{the price of anarchy} can be called as \textit{the price of cognition}. This represents \textit{the most unfavorable ratio between the overall impairment caused by toxic neuropeptide's (neurotransmitter's) equilibrium in comparison to the optimal state of synapses (synapses optimum)}. To put it differently, \textit{the price of cognition} measures \textit{the loss of cognitive ability resulting from increased concentrations of toxic neuropeptides/neurotransmitters}.  Finally, a replicator equation was proposed within this framework that leads to the establishment of the synapses optimum during the neurotransmission process.  Our model represents a high-level simplification and abstraction of the natural neurotransmission process involving two neurons. Nevertheless, we envisioned that this mathematically abstract model could serve as a source of motivation to instigate novel experimental, mathematical, and computational research avenues in the field of contemporary neuroscience.

\section*{Declaration of competing interest}
The authors declare that they have no known competing financial interests or personal relationships that could have appeared to influence the work reported in this paper.

\section*{Declaration of data availability statement}
The authors declare that data sharing is not applicable to this article as no new data were created or analyzed in this study.

\end{document}